\documentclass[pdflatex,sn-mathphys]{sn-jnl}

\usepackage{enumerate}
\usepackage{soul}
\usepackage{makecell}
\usepackage{subfigure}
\usepackage{amsmath}
\usepackage{graphicx}
\usepackage{float} 
\usepackage{lmodern}
\usepackage{type1cm}
\usepackage{fix-cm}
\usepackage{mathrsfs}
\usepackage{anyfontsize}

\DeclareSymbolFont{rsfs}{U}{rsfs}{m}{n}
\DeclareSymbolFontAlphabet{\mathscr}{rsfs}

\usepackage{graphicx,xcolor,float}
\jyear{2021}%

\theoremstyle{thmstyleone}%
\newtheorem{theorem}{Theorem}
\theoremstyle{thmstyletwo}%

\theoremstyle{thmstylethree}%

\begin{document}

\title[DHSA]{DHSA: Efficient Doubly Homomorphic Secure Aggregation for Cross-silo Federated Learning}


\author[1,3]{\fnm{Zizhen} \sur{Liu}}\email{liuzizhen18s@ict.ac.cn}
\author[2]{\fnm{Si} \sur{Chen}}\email{si.chen@osr-tech.com}

\author*[1,3]{\fnm{Jing} \sur{Ye}}\email{yejing@ict.ac.cn}

\author[2]{\fnm{Junfeng} \sur{Fan}}\email{fan@osr-tech.com}
\author[1,3]{\fnm{Huawei} \sur{Li}}\email{lihuawei@ict.ac.cn}
\author[1,3]{\fnm{Xiaowei} \sur{Li}}\email{lxw@ict.ac.cn}

\affil*[1]{\orgdiv{Institute of Computing Technology}, \orgname{Chinese Academy of Sciences}, \orgaddress{\street{No.6 Kexueyuan South Road}, \city{Beijing}, \postcode{100190}, 
\country{China}}}

\affil[2]{
\orgname{Open Security Research}, \orgaddress{\street{No.18 Science and technology Road}, \city{Shenzhen}, \postcode{518063}, 
\country{China}}}

\affil[3]{
\orgname{CASTEST}, \orgaddress{\street{No.18 Zhongguancun Road}, \city{Beijing}, \postcode{100083}, 
\country{China}}}


\abstract{Secure aggregation is widely used in horizontal Federated Learning (FL), to prevent leakage of training data when model updates from data owners are aggregated. Secure aggregation protocols based on Homomorphic Encryption (HE) have been utilized in industrial cross-silo FL systems, one of the settings involved with privacy-sensitive organizations such as financial or medical, presenting more stringent requirements on privacy security. However, existing HE-based solutions have limitations in efficiency and security guarantees against colluding adversaries without a Trust Third Party. 

This paper proposes an efficient Doubly Homomorphic Secure Aggregation (DHSA) scheme for cross-silo FL, which utilizes multi-key Homomorphic Encryption (MKHE) and seed homomorphic pseudorandom generator (SHPRG) as cryptographic primitives. The application of MKHE provides strong security guarantees against up to $N-2$ participates colluding with the aggregator, with no TTP required. To mitigate the large computation and communication cost of MKHE, we leverage the homomorphic property of SHPRG to replace the majority of MKHE computation by computationally-friendly mask generation from SHPRG, while preserving the security. Overall, the resulting scheme satisfies the stringent security requirements of typical cross-silo FL scenarios, at the same time providing high computation and communication efficiency for practical usage. We experimentally demonstrate our scheme brings a speedup to 20$\times$ over the state-of-the-art HE-based secure aggregation, and reduces the traffic volume to approximately 1.5$\times$ inflation over the plain learning setting.}




\maketitle
\section{Introduction}\label{sec1}

Recently, federated learning (FL) has emerged as a popular solution to build machine learning models based on data sets distributed across multiple parties \cite{RN69}. Due to concerns of data privacy, participants in FL (e.g. mobile devices or whole organizations) collaborate to train a machine learning model without sharing their raw training data \cite{RN31}. Instead, in each epoch of the training process, the model is first trained locally on respective training data, and the local model update information is exchanged among the participants to generate a global model. According to application scenarios \cite{RN68}, FL can be divided into cross-device FL and cross-silo FL. Based on the distribution characteristics of the data sets, FL can be categorized into vertical FL and horizontal FL. This paper focuses on the cross-silo, horizontal FL, where a small number of organizations (e.g. medical or financial institutions)
with reliable communications and relative abundant computing
power jointly build a model by updating the global model with the aggregation of local model updates.

However, as recent studies argue, although the local training data is not exposed in FL, the exchanged model data still carries extensive information about the training data. The inference attacks can be potentially exploited to reveal sensitive information about the training data, which may occur during the training process or upon the trained model \cite{8835245, 9148790, 274683}. The inference attacks during the training process refer to the information inferring given the model update. For example, Deep Leakage from Gradients (DLG) algorithm can completely recover the training data from the uploaded gradients during the training process \cite{RN34}. Trained models obtained by FL can also leak private information about the training data, which can be handled separately.  Consequently, privacy protection techniques including secure multiparty computation (SMC) \cite{RN44, RN167, RN168, RN169, RN170}, homomorphic encryption \cite{RN48, RN47, RN27, RN71, DBLP:conf/ndss/SavPTFBSH21}, and differential privacy \cite{RN49, RN77} are proposed to implement secure aggregation schemes to address such indirect information leakage issues.

Existing works on privacy-preserving FL have mainly focused on cross-device settings, where the clients are up to millions of mobile or IoT devices with limited computing power and unreliable communications. In this scenario, the main challenges are the efficiency of systems with a large number of participants and the resilience against the frequent dropout problem, which are not the dominant concerns in the cross-silo setting \cite{RN68}. Instead, in cross-silo FL, the participating institutions such as medical \cite{RN59} or financial \cite{RN68} put forward more stringent requirements on privacy. On one hand, the model update and final trained model should be exclusively released to only those participating organizations, preferably not to external parties including the central server and any Trust Third Party. On the other hand, in the colluding case, the local model update is desired to be exchanged without exposure. Efficiency is another concern: computation and communication costs should be low enough to enable industrial deployment. Besides, the trained model is expected to achieve the same quality as one obtained from FL with plaintext learning. 

To our knowledge, there exist only a few solutions for secure aggregation of cross-silo FL settings, where the additively Homomorphic Encryption (HE) is commonly used for the strong privacy guarantee. With HE, clients encrypt their updates in such a way that the server can perform aggregation directly on ciphertexts without deriving any meaningful information about the plaintexts \cite{RN47}. However, for deep neural networks with millions of parameters, direct encryption of gradients brings significant overhead to computation and communication. One of the state-of-art works, BatchCrypt, develops batch encryption techniques enabling data-parallel computation and data compression to improve efficiency \cite{RN48}. While the computation overhead is amortized to some extent, the computation time of aggregation is still orders of magnitude larger than plain aggregation without encryption, and communication traffic is inflated by at least 2 times \cite{RN48}. More seriously, the security guarantee breaks down against collusion threats. Although threshold variant of homomorphic encryption \cite{RN27} and multi-input functional encryption \cite{RN71} are proposed to handle the colluding threats, these solutions require a Trust Third Party, which usually is not practical in cross-silo FL.

To address the limitations of security and efficiency as mentioned above, we develop a practical secure aggregation scheme, named doubly homomorphic secure aggregation (DHSA), which provably meets the requirements of cross-silo FL. DHSA comprises two protocols: the Homomorphic Model Aggregation (HMA) protocol and the Masking Seed Agreement (MSA) protocol. In the HMA Protocol, instead of encrypting all the model data with HE, a simple masking scheme based on Seed Homomorphic Pseudorandom Generator (SHPRG) is utilized to hide the model updates. The masks are generated by SHPRG, taking locally sampled masking seeds as inputs, and the demasking seed of the aggregation results is computed securely by the MSA protocol, where multi-key Homomorphic Encryption (MKHE) is utilized to ensure the aggregation of masking seeds is learned only by the clients, and nothing beyond the aggregation is revealed. Based on the basic protocols, instead of calling MSA protocol to return one demasking seed each time the model aggregation is performed by HMA protocol, we optimize the efficiency by setting up demasking seeds for multiple epochs during per execution of MSA protocol, which takes advantage of the packing technique of selected MKHE to further reduce overheads.

In summary, this simple construction has the following merits compared with other secure aggregation schemes for cross-silo FL:
\begin{enumerate}[(1)]
\item Securely perform the model update aggregation. The combination of individual masking encryption and MKHE provably provides strong privacy security even when the server colludes with up to $N-2$ clients. DHSA provides strong security guarantees without the assumption of TTP, making the implementation more practical than existing works.  
\item Efficient in both computation and communication aspects. The masking operation is simple to implement, and we take advantage of the homomorphic property of SHPRG to simplify the achievement of demasking further, enhancing the efficiency from both communication and computation aspects compared with the previous HE-based solutions. DHSA brings $O(M)$ computation complexities for each client and the server, where $M$ is the number of model parameters, and only approximately 1.5× inflation of traffic. We experimentally demonstrate DHSA brings a speedup to 20× over the latest related work, BatchCrypt.
\item No compromise on accuracy. DHSA experimentally obtains a similar accuracy than non-secure uncompressed FedAvg.
\end{enumerate}

\section{Background and Problem Statement}\label{sec2}

\subsection{Privacy Threats in Federated Learning}
The basic concept of federated learning is to protect the privacy of training data by keeping the training data on local devices. However, the shared model-related information still carries extensive information about the corresponding training data, resulting in the exposure of training data privacy from model inference attacks. Depending on the data exploited by attackers to infer privacy, the model inference attacks may occur at different stages of FL learning. Also, the trained model should be protected for its intrinsic commercial value. Overall, the following three types of privacy 
concerns can be imposed. 
\begin{enumerate}[(1)]
\item Inference attack during the learning process. The information flow of intermediate results in the learning process mainly includes individual model updates and the aggregation of clients' model data. The individual model updates (model parameters or gradients) can leak sensitive information about the individual training data. For example, various membership inference approaches are introduced to verify if a record is in the training dataset based on the parameter updates  \cite{Nasr2018ComprehensivePA, 8737416, 9148790}; property inference attacks \cite{8737416} allow the adversaries to learn sensitive properties of individual training data; more seriously, DLG algorithm \cite{RN34}\cite{zhao2020idlg} is presented most lately to completely steal the training data from gradients. The aggregation of clients' model data is the merge of multiple contributions, spontaneously covering up the sensitive information of a single party, with relatively low risk to reveal the individual information. 
\item Inference attack based on the final model outputs. The output of the trained model is another attack source to reveal training data privacy. Reference \cite{274683} shows even when parties obtain only black-box access to the final model, an attacker can infer a dataset property based on a set of queries.

\item In addition to preventing training data from leaking through inference attacks, the device or organization may deploy the trained machine learning model for the financial application, and restrict the ability to inspect, misuse or steal the model. 
\end{enumerate}

 \subsection{Privacy Solution in Federated Learning}
Existing privacy solutions for FL apply privacy protection techniques, including Secure Multiparty Computation (SMC) \cite{RN44, RN167, RN168, RN169, RN170}, Homomorphic Encryption (HE) \cite{RN48, RN47, RN27, RN71,DBLP:conf/ndss/SavPTFBSH21}, and Differential Privacy (DP) \cite{RN49, RN77} to address the inference attacks at different stages. In this section, we briefly review these strategies, and discuss the applicability for cross-silo FL. 

Secure Multi-party Computation (SMC) guarantees that a set of parties compute a function in a way that each one cannot learn anything except the output, and SMC-based techniques mainly utilize Yao’s garbled circuits \cite{4568388} or secret sharing \cite{10.1145/359168.359176}\cite{RN170}\cite{10.1007/3-540-45472-1_7}. A notable work of SMC-based privacy-preserving FL is the secure aggregation protocol proposed by Bonawitz et al. \cite{RN44}. They develop a double masking solution (Double-mask), which achieves secure aggregation against colluding participants and is robust to dropout. In this solution, for any pair of clients, they securely agree upon a shared mask, and each client generates an individual mask. The message of each client is masked with $N$ masks in a designed way. By secretly sharing all the masking seeds, their reconstruction is allowed upon sufficient shares. Thus, Double-mask is robust to dropouts frequently occurring in cross-device FL. However, the quadratic growth of computation overhead w.r.t. $N$ is the major bottleneck. Subsequently, TurboAgg \cite{RN167} utilizes a circular communication topology to reduce the computation overhead. SecAgg+ achieves polylogarithmic communication and computation per client via communication graph \cite{RN169}. FastSecAgg presents an FFT-based multi-secret sharing scheme to obtain $O(M\log N)$ cost \cite{RN168}. Although gaining the improvement of efficiency, the security guarantee of those works is weaker to some extent. In the cross-silo FL, unlike the massive-scale user case, all clients are able to participate in each round with robust connectivity to the system. The above methods that handle the dropout problems at high overhead are not the best choice for cross-silo settings. In addition, the sacrifice of security is not acceptable in cross-silo FL. For example, TurboAgg provides the privacy guarantee against up to $N/2$ colluding parties, which is insufficient for the security requirements of cross-silo FL.

Apart from SMC, HE is another common tool for addressing the privacy leakage problem during the learning process of FL, especially of cross-silo FL, because it seeks stronger guarantees on privacy. Many recent works \cite{RN48, RN47, RN27, RN71, RN64} advocate the use of additively HE schemes, notably Paillier \cite{RN26}, as the primary means of privacy guarantee in FL. In this scheme, all the model gradients updated by clients are encrypted by HE to ensure no information leakage, and the homomorphic property of HE allows the aggregation operation to be performed directly on ciphertexts without prior decryption. However, HE performs complex cryptographic operations that are expensive to compute. Especially for deep neural networks with millions of parameters, direct encryption of gradients brings significant overhead to computation and communication. To improve the efficiency of Paillier-based secure aggregation, BatchCrypt \cite{RN48}  develops new quantization and encoding schemes. This method achieves batch encryption which enables data-parallel computation and data compression. Although compared with Paillier without batching, the computation and communication overheads are significantly reduced, questions arise about the collusion threats. For all clients utilizing the same secret and public key pairs, it’s easy to decrypt the ciphertext of others when a client colludes with the central server. Recent works \cite{RN27, RN71} handle the colluding problem by developing the distributed decryption method. In these schemes, different public keys are utilized to encrypt the model data, and the decryption needs the collaboration of multiple parties. Reference \cite{RN27} applies the threshold variant of Paillier cryptosystem, where a Trust Third Party distributes shares of the secret key to the set of participants such that no subset of the parties smaller than a predefined threshold is able to decrypt values. However, it is hard to find a TTP in the cross-silo setting actually. From another point of view, the introduction of TTP is an extra threat to privacy security. Multi-party homomorphic encryption has been also utilized in secure aggregation of FL \cite{DBLP:conf/ndss/SavPTFBSH21, froelicher2021scalable} to provide security against colluding parities, nevertheless, resulting in unacceptable data transmission cost. 

Differential Privacy (DP) improves the privacy of machine learning models by injecting noises \cite{RN49, RN77}. Wei et al. \cite{wei2020federated} propose a framework based on the concept of DP to prevent information leakage from model updates. DP can also be combined with secure aggregation to defense the inference attacks from trained models \cite{RN71}, or further ensure that the aggregated result does not reveal additional information to the server \cite{7286780}. However, the application of DP always needs to tradeoff between the model quality and privacy protection levels. It's challenging for DP itself to provide the required security guarantees of cross-silo FL while preserving the model quality. It can be applied in our scheme as an independent supplementary component to provide further security guarantees. 
\subsection{Problem statement}

In this paper, we target privacy of the horizontal, cross-silo federated learning, where $N$ data owners (also called clients) collaboratively train a model with $M$ parameters with the coordination of an aggregator (also called the server). The typical optimization algorithms Federated Averaging is leveraged to optimize the model \cite{article}, where each client provides the local model update $m_u$ and the aggregation result $\sum m_u$ is computed to update the global model. This process is repeated until the model accuracy and loss reach convergence to obtain the final global model $m_{\textrm{final}}$. Then the trained model is released to the participating data owners in cross-silo settings.  To avoid the inference of training data privacy from the exchanged model update $m_u$ during the learning process \citep{274683}, secure aggregation aims to compute $\sum{m_u}$ without revealing additional sensitive information beyond the model aggregation.

The threat model is honest-but-curious, and allows colluding. The potential adversaries in FL may be clients or the server who have access to the intermediate exchanged information. In the colluding case, the adversary may control a set of up to $T$ clients, and may also control the aggregator, where the adversary can view the internal state and all the messages received/sent by the controlled ones in colluding sets. The adversary attempts to learn about the model updates from other parties by using the viewed messages exchanged during the execution of the protocol. As a result, adversaries steal the valued model or utilize the model data to infer the training data of some clients.  

With the above threats under consideration, the goal of the proposed secure aggregation is to meet the standards of cross-silo FL mainly from the following three aspects.
\begin{enumerate}[(1)]
\item Security.
As opposed to the cross-device FL setting, the cross-silo FL setting typically requires much more stringent privacy protections for the training data, and for the trained model \cite{RN68}. The security requirements are reflected in the following aspects. (1) A secure solution should preferably achieve the highest possible value of the collusion threshold (the maximum number of colluding parties allowed within the security guarantee), i.e. $T_{col}=N-2$. It means that clients and the server learn global model updates without revealing any additional information (e.g. other parties’ local model update or training data) even in the worst colluding cases where the server colludes with $N-2$ clients. \footnote{Note that the security against collusion of $N-1$ parties is not considered in the aggregation scheme, because even in the ideal secure aggregation system, the model data $m_{h}$ of the honest client can be indicated from $\sum{m_u}$ and the joint view of $N-1$ colluding parties.}  (2) The exchanged model data should only be visible to the necessary parties involved. Each local model update $m_u$ can be accessed only by data owner $u$. The aggregated intermediate model $\sum{m_u}$ and the final trained model $m_{\textrm{final}}$ can be accessed only by data owners, and in particular, not accessible by the server or any Trust Third Party.
\item Efficiency.
The privacy-preserving scheme should be able to be implemented efficiently to enable the industrial deployment of cross-silo FL.
The computation and communication overhead should be reduced to minimize the communication traffic and speed up the learning process.
\item Model quality.
A strong privacy guarantee should not be achieved at the cost of model quality. The degradation of the trained model’s accuracy should be minimal compared with the ideal unquantized model trained without a secure aggregation scheme.
\end{enumerate}
\section{Cryptographic Tools}\label{sec3}

\subsection{Seed Homomorphic PRG}
Recall that a pseudorandom generator (PRG) is a deterministic polynomial-time algorithm $F:\{0,1\}^n\rightarrow \{0,1\}^m$ such that $n<m$, and for randomly distributed $s\in \{0,1\}^n$ and $r\in \{0,1\}^m$, the distributions of $F(s)$ and $r$ are computationally indistinguishable. A PRG $F:\chi \rightarrow \gamma $, where $(\chi,\oplus)$ and $(\gamma, \otimes )$ are 
groups, is said to be seed homomorphic if the following property holds \cite{RN39}: 
For every $s_1, s_2\in \chi $, we have that $F(s_1)\otimes F(s_2)=F(s_1\oplus s_2)$.

An almost Seed Homomorphic Pseudorandom Generator(SHPRG ) can be constructed based on the Learning With Rounding (LWR) problem. With the public parameters $n,m,p,q$ satisfying $p<q, n<m$, the SHPRG $G(s): \mathbb{Z}_{q}^{n} \rightarrow \mathbb{Z}_{q}^{m}$ can be defined as $G(s)=\left \lceil A^{T}\cdot s\right \rfloor_p$, where $A$ is another public parameter randomly sampled from $\mathbb{Z}_{q}^{n\times m}$, $s$ is uniform in $\mathbb{Z}_{q}^{n}$, and $\lceil\cdot\rfloor_p$ is defined as ${\left \lceil x \right \rfloor}_p=\left \lceil x\cdot p/q \right \rfloor$ for $x\in \mathbb{Z}_q$. It is almost seed homomorphic in the following sense: 
$$
G(s_1+s_2)=G(s_1)+G(s_2)+e, e\in [-1, 0,1]^{m}.
$$

Note that the security of the above SHPRG depends on the hardness of LWR$_{n,q,p}$ problem \cite{RN76}. The value of $1/p$ is proportional to the error rate $\alpha$ in Learning With Error (LWE), so the selection of parameters should assure that LWE$_{n,q,1/p}$ has difficulty satisfying the security level objective. 

Multiple privacy-critical applications have been built from Seed Homomorphic PRG or the related preliminary Key Homomorphic Pseudorandom Functions, such as distributed PRFs, undatable encryption \cite{RN39} and private stream aggregation \cite{RN171, RN172}. The homomorphism property is in support of specific applications with provable security.

\subsection{Multi-key Homomorphic Encryption}

Homomorphic Encryption (HE) schemes allow certain mathematical operations to be performed directly on ciphertexts, without prior decryption. Exactly, the homomorphic scheme allows some operations to be directly performed on the ciphertexts $E(m_1)$ and $E(m_1)$ to obtain the result which corresponds to a new ciphertext whose decryption yields the sum or the multiplication of the plaintext $m_1$ and $m_2$. The Paillier \cite{RN48}, BFV \cite{cryptoeprint:2012:144} and
CKKS \cite{cheon2017homomorphic} schemes are the most prevalent HE variants. The Paillier
scheme allows only addition operation to be performed on ciphertexts,
while the BFV and CKKS schemes permit both additions and
multiplications. The ciphertext packing technique of BFV and CKKS allows encrypting multiple data in a single ciphertext and performing parallel homomorphic operations in a SIMD manner \cite{2014Fully}. In the traditional single-key HE schemes, all the parties involved share the same key for encryption and decryption, and the decryption can be done independently at any party.

Multi-key Homomorphic Encryption (MKHE) is a cryptosystem in which each party generates its own keys, and the specific operation can be performed on ciphertexts encrypted by different parties. The decryption of the obtained new ciphertexts is achieved by combining the respective secret keys associated with these ciphertexts \cite{10.1007/978-3-662-53644-5_9}. MK-BFV is a multi-key variant of the BFV HE scheme \cite{cryptoeprint:2012:144} which is an Ring Learning With Error (RLWE)-based  \cite{10.1145/1568318.1568324} cryptosystem (refer to Appendix A and B for the details of RLWE and BFV). So far, there are at least two types of MK-BFV schemes: MK-BFV based on ciphertexts extension (see Appendix C) \cite{RN81} and Compact MK-BFV \cite{RN43}.  Since homomorphic multiplication is not involved in the proposed protocol, only the linear version of the two types of MK-BFV schemes is stated in this paper.

Let $\mathcal{P}$ be a set of $N$ parties that have access to an authenticated channel and a random common reference string
(CRS), the Compact MK-BFV scheme is consisted of a tuple of steps (KeyGen, Enc, Dec, Eval, Public-key-switching) based on BFV scheme, which are defined as follows.  
\begin{enumerate}[\textbullet]
\item Setup: $pp \leftarrow \textrm{MKBFV.Setup}(\lambda)$. $\textrm{MKBFV.Setup}$ takes the security and homomorphic capacity parameters as inputs, setting the RLWE dimension $n$, ciphertext modulus $q$, key distribution $\chi $ over $R$, and error distribution $\psi,\phi$ over $R_q$. With uniformly sampled $a\leftarrow R_q$, the outputs is a public
parameter $pp=\{n, q, \chi, \psi,\phi,  a\}$, which is an implicit argument to the other procedures.
\item Key Generation: $(\textrm{sk}_i, \textrm{cpk}) \leftarrow \textrm{MKBFV.} \prod_{\textrm{KeyGen}}(pp)$.\\ $\textrm{MKBFV.} \prod_{\textrm{KeyGen}}(pp)$ is a protocol whose input is the public parameter $pp$, and outputs are key pairs ($\textrm{sk}_i, \textrm{cpk}$) to each party $P_i$. The protocol is composed of two steps:

\fbox{%
\parbox{0.88\textwidth}{%
    \begin{center}
    \centering{\textbf{Key Generation Protocol.}}
    \leftline{\textbf{Input}: Public parameter $pp$;}
    \leftline{\textbf{Output}: $\{\textrm{sk}_i, \textrm{cpk}\}$ for each party $P_i$;}
    \leftline{\emph{Client $P_i$}:}
    \leftline{\shortstack[l]{$\{\textrm{sk}_i,\textrm{pk}_i\}\leftarrow \textrm{BFV.KeyGen}(pp)\}$. Generates $\{\textrm{sk}_i,\textrm{pk}_i\}$ by the usual\\ key generation calculation of BFV, where $\textrm{sk}_i=s_i, \textrm{pk}_i=(p_{0,i}, p_{1})$ \\ $=(-s_i\cdot a+e_i, a)$.}}
    \leftline{\emph{Out}:} \leftline{\shortstack[l]{$\textrm{cpk}\leftarrow\textrm{MKBFV.ComKey(}\textrm{pk}_1, \textrm{pk}_2...\textrm{pk}_N)$. Given $\{\textrm{pk}_i\}$ of all parties, \\the common public key is computed by cpk = $([\sum_{P_i\in \mathcal{P}}{p_{0,i}}]_q, p_1)=$\\$([-(\sum_{P_i\in \mathcal{P}}{s_i})a +\sum_{P_i\in \mathcal{P}}{e_i}]_q, a)$, and the corresponding secret key \\is $\textrm{sk}=\sum_{P_i\in \mathcal{P}}{s_i}$.}}
    \end{center}
    }%
}

\item Encryption: $\textrm{ct}_i\leftarrow \textrm{MKBFV.Enc}(\textrm{cpk}, x_i$). Upon receiving the common public key cpk and a plaintext $x_i$, $\textrm{MKBFV.Enc}(\textrm{cpk}, x_i)$ exploits the usual encryption calculation of BFV to encrypt message under $\textrm{sk}$ and outputs $\textrm{ct}_i$ = BFV.Enc($\textrm{cpk}, x_i)\in R_q^2$.
\item Decryption: $x\leftarrow \textrm{MKBFV.Dec(\textrm{sk}, \textrm{ct})}$. Given a ciphertext $\textrm{ct}$ encrypting $x$ and the corresponding secret key $\textrm{sk}$, $\textrm{MKBFV.Dec}$ runs the usual decryption calculation of BFV to output $x = \textrm{BFV.Dec(\textrm{sk}, \textrm{ct})}$.
\item Evaluation: ${\textrm{ct}}'\leftarrow\textrm{MKBFV.Eval(ct}_1,\textrm{ct}_2$, $F(\cdot)$). Given the ciphertexts $\textrm{ct}_1, \textrm{ct}_2$, as well as a linear function $F(\cdot)$, the evaluate function $\textrm{BFV.Eval}$ outputs the ciphertext $\textrm{ct}'$ such that $\textrm{BFV.Dec(sk, ct}')=F(x_1, x_2)$.
\item Public-key-switching: ${\textrm{ct}}'\leftarrow  \prod_{\textrm{PubKeySwitch}}(\textrm{ct},{\textrm{pk}}', \textrm{sk}_1,...,\textrm{sk}_N$).\\ $\prod_{\textrm{PubKeySwitch}}(\textrm{ct},{\textrm{pk}}', \textrm{sk}_1,...,\textrm{sk}_N)$ is a protocol where the participants collaboratively re-encrypt the ciphertexts without decrypting them. Given a ciphertext $\textrm{ct}=(c_0, c_1)$ under $\textrm{sk}=\sum{\textrm{sk}_i}$ and an output public-key ${\textrm{pk}}'=(p_0',p_1')$ for secret-key ${\textrm{sk}}'$, the parties re-encrypt $\textrm{ct}$ under ${\textrm{sk}}'$ to obtain ${\textrm{ct}}'=(c_0',c_1')$. The obtained ciphertexts can be decrypted by performing BEV.Dec(${\textrm{sk}}', {\textrm{ct}}'$).
\end{enumerate}

\fbox{%
\parbox{0.88\textwidth}{%
    \begin{center}
    \centering{\textbf{Public Key Switch Protocol.}}
    \leftline{\textbf{Public input}: $\textrm{ct}=(c_0, c_1)$, ${\textrm{pk}}'=(p_0',p_1')$;}
    \leftline{\textbf{Private input}: $\textrm{sk}_i$ for party $P_i$;}
    \leftline{\textbf{Output}: ${\textrm{ct}}'=(c_0',c_1')$;}
    \leftline{\emph{Client $P_i$}:}
    \leftline{\shortstack[l]{$(h_{0,i}, h_{1,i})\leftarrow \textrm{MKBFV.ParKeySw}({\textrm{pk}}', \textrm{sk}_i, \textrm{ct})$. Samples $u_i\leftarrow \chi$, \\$e_{0,i}\leftarrow \phi$ and computes $(h_{0,i}, h_{1,i})= (s_ic_1+u_ip_0'+e_{0,i}, u_ip_1'+e_{1,i})$.}}
    \leftline{\emph{Out}:} \leftline{\shortstack[l]{$\textrm{ct}'\leftarrow\textrm{MKBFV.MerKeySw}(\{(h_{0,i}, h_{1,i})\}_{i\in\mathcal{P}})$. After calculating \\$h_0=\sum_{j}{h_{0,j}}$ and $h_1=\sum_{j}{h_{1,j}}$ upon the given $\{(h_{0,i}, h_{1,i})\}_{i\in\mathcal{P}}$, \\$\textrm{MKBFV.MerKeySw}$ outputs the re-encryption ciphertexts $\textrm{ct}'=$\\$(c_0', c_1')=(c_0+h_0, h_1)$.}}
    \end{center}
}%
}

The Compact MK-BFV satisfies the semantic security and can provide security guarantees against $N-1$ colluding adversaries. In this paper, we implement secure aggregation by exploiting Compact MK-BFV instead of that based on ciphertext extension. We will explain the reason in the next section.

\section{Doubly Homomorphic Secure Aggregation Scheme}\label{sec4}

In this section, we describe our proposed doubly homomorphic secure aggregation scheme (DHSA) for cross-silo FL which includes two layers of protocols: the Homomorphic Model Aggregation (HMA) protocol for model aggregation and the Masking Seed Agreement (MSA) protocol for masking seeds aggregation. We begin with the motivation and observation of our design. Then we describe the overall process of the DHSA scheme. Finally, we describe the details of the HMA protocol and the MSA protocol.
\subsection{Observation} 
To meet the strict requirements for privacy of cross-silo FL, we propose the application of MK-BFV. To overcome the efficiency bottleneck, SHPRG is leveraged in our framework, which also preserves the security property. In detail, we first show why we apply MK-BFV to cross-silo FL and choose the Compact version. We then explain why to introduce SHPRG-based HMA protocol. 
\subsubsection{Why utilize Compact MK-BFV?}
The first concern of our solution is privacy security, which is especially critical to cross-silo FL. The main security goals are as follows: (1) achieve the $N-2$ collusion threshold. (2) ensure that only the data owners have access to the aggregated intermediate model update and the final trained model. (3) provide the privacy guarantee obviating the need for a TTP. We propose a solution utilizing the multi-key variants of BFV to simultaneously address these challenges. Different from the threshold Paillier utilized in previous work \cite{RN27}, the distributed encryption and decryption keys in MK-BFV are generated individually or jointly among participants without the assistance of the TTP. Besides, the corresponding ciphertexts can be successfully decrypted only when all related parties share the partial computation results. Thus, in general, up to $N-1$ colluding parties can be tolerated without the assumption of TTP, and for the aggregating operation in FL, up to $N-2$ colluding parties can be tolerated. 

As a preparation, we investigate the applicability of the two mainstream constructions of MK-BFV in secure aggregation. One of the choices is the construction based on ciphertexts extension \cite{RN81}. We found it infeasible to apply it in secure aggregation because the curious server has the ability to learn the individual model data. To achieve secure aggregation based on this type of MK-BFV, the following steps should be taken. Each client $P_i$ encrypts individual model updates with the individual public key $\textrm{pk}_i$ to obtain encrypted model data $\textrm{ct}_i=(c_{0,i}, c_{1,i})$ and uploads it to the server for aggregation. The server extends all ${\textrm{ct}_i}$'s to ${\overline{\textrm{ct}_i}}\textrm{'s}\in R_q^{N+1}$ before summing them up. As discussed in Appendix C, for addition arithmetic, the summed ciphertexts associated to $N$ different parties have the form $\overline{{\textrm{ct}}'}=(\sum_{i=1}^N{c_{0,i}},c_{1,1},c_{1,2},...c_{1,N})\in R_q^{N+1}$. To decrypt the result, $\left \langle \overline{{\textrm{ct}}'},(1,s_1,...,s_N) \right \rangle$ needs to be calculated. While for FL scenarios, all clients perform distributed decryption under the coordination of the server, which consists two phases: first, each client $P_i$ provides the partial decryption shares $\mu_i=\overline{ct'}[i]\cdot s_i+e_i=c_{1,i}\cdot s_i+e_i(\mod q)$ and uploads it to the server; then the server merges all received ${\mu_i}$’s by computing $\mu = \overline{ct'}[0]+\sum_{i=1}^N{\mu_i}$.

Problem arises now that both the raw ciphertexts $\textrm{ct}_i=(c_{0,i}, c_{1,i}) $ before being extended and the partial decryption share $\mu_i$ of client $i$ can be accessed by the server. The curious server can retrieve the individual model data by computing $c_{0,i}+\mu_i$, which seems inevitable due to only simple addition operation is involved in the FL systems considered.  

Instead, we invoke the Compact MK-BFV scheme, where a common public key is set up collaboratively to encrypt individual model updates. By doing so, all ciphertexts, including the individual model update and the aggregation, can only be decrypted jointly by all related participants. The aggregation of ciphertexts is done in accordance with the conventional ciphertexts additive operation without extension, and the partial decryption result combines the sum over all ciphertexts as well as the individual secret key. Even when the server gets access to the individual ciphertexts and all the partial decryption results, the server cannot decrypt the individual ciphertexts. 

Further, to meet the request of cross-silo FL that the trained model and aggregated intermediate model updates should be released to no external parties, including the server, the Public Key Switch (PKS) protocol in Compact MK-BFV allows the final decryption to be performed by clients. As opposed to the two-step decryption in MK-BFV, which causes the plaintext to leak to the server, the PKS protocol in Compact MK-BFV adjusts the partial decryption by adding a new LWE instance and returns a new ciphertext instead of the plaintexts after merging. The new ciphertexts can be decrypted with the specific secret key, which is kept private by the clients. 

\subsubsection{Why propose SHPRG-based Homomorphic Model Aggregation protocol?}
Although MK-BFV-based protocol already provides a secure aggregation solution meeting the security requirements, directly applying it on the model updates with millions of entries induces prohibitively large computation and communication overhead. Therefore, we introduce SHPRG to perform the encryption and decryption of the model update instead of MK-BFV, which is the bulk of the computation, while only the seed of SHPRG is aggregated using MK-BFV. We leverage the following traits of SHPRG to reduce the overheads:
(1) the almost additive homomorphic property $\sum_iG(x_i)\approx G(\sum_i x_i)$. 
(2) efficient to compute compared with MK-BFV operations.
(3) the size of inputs (seed) is much smaller than the size of outputs. 

Explicitly, data owner $u$ encrypts its local modal update by
$$
y_u=x_u+G(k_u),
$$
and the summed model update is decrypted by
$$
\sum_u x_u\approx\sum_u y_u-G(k_0),
$$
where $k_u$ is the SHPRG seed generated by data owner $u$, and $k_0=\sum_u k_u$ is the aggregated seed from MK-BFV-based protocol. In this way, the strong security guarantee from MK-BFV can be preserved, and significant efficiency improvement is gained from using SHPRG.

\subsection{Overview of Doubly Homomorphic Secure Aggregation Scheme}

As discussed above, DHSA is constructed with two layers of protocols: the model updates are securely shared and aggregated following the Homomorphic Model Aggregation (HMA) protocol, which calls the Masking Seed Agreement (MSA) protocol to return the demasking seed to the clients and enable demasking to obtain the global model update. Figure~\ref{FIG:1} depicts the mechanism of the combination of the two protocols, which guides the following design of DHSA. The most straightforward construction to build secure aggregation is that, for every epoch $t$ where HMA is performed to do aggregation of model updates, MSA is called once to provide the current demasking seed $k_0^t$. Recall that in cross-silo scenarios, all participating clients in $\mathcal{C}_N$ remain available during the whole iterative learning process of FL, which means the model aggregation in the HMA protocol for every epoch is over the same set of clients. As a result, the demasking seeds $k_0^t$'s for all involved epochs correspond to the same set of clients. Thus, masking and demasking seeds for multiply epochs utilized in the HMA protocol can be prepared in advance. Figure~\ref{FIG:2} depicts the overall process of the mechanism. The MSA protocol is firstly carried out to set up the masking seed pairs for $\tau$ epochs, which are consumed in the next $\tau$ rounds of the HMA protocol. If the FL training convergences within $\tau$ epochs, the execution is ended. If not, all the participants call the MSA protocol another time and repeat the steps until the convergence is achieved.
\begin{figure}
	\centering
	\includegraphics[scale=0.4]{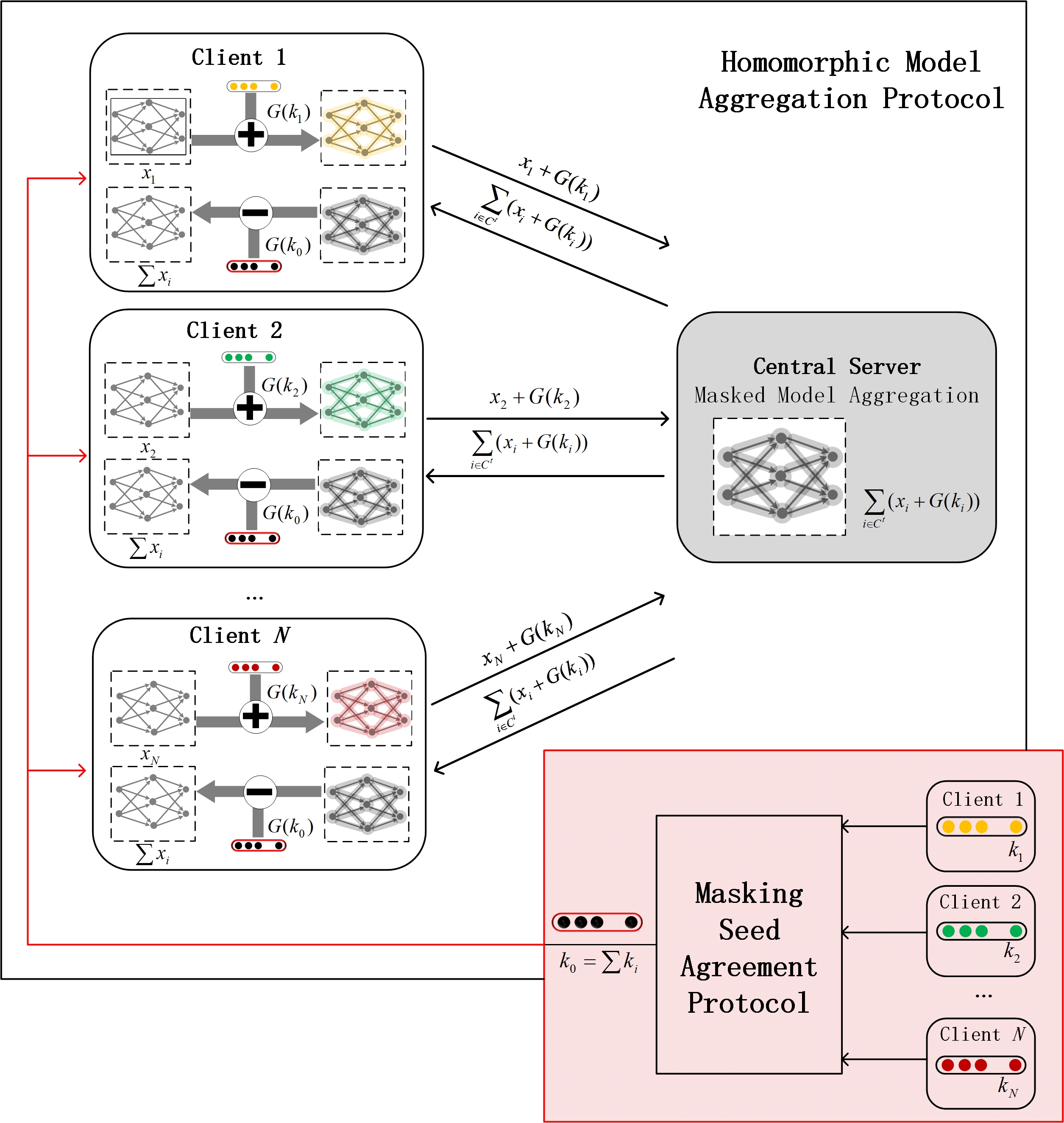}
	\caption{The combination mechanism of two layers of protocols.}
	\label{FIG:1}
\end{figure}

This construction can be implemented more efficiently because we take advantage of the SIMD technique of BFV by packing multiple masking seeds in one BFV ciphertext. In addition to the SIMD optimization, we can evaluate multiple ciphertexts in one execution of the MSA protocol. As a result, only a single round of communication is required for each learning epoch of FL to perform the HMA protocol. It's sufficient to perform the MSA protocol once every $\tau$ epochs, which significantly reduces the number of rounds of communication. For the FL process with a total of $T$ epochs, if $T$ is not bigger than $\tau$, the MSA protocol is executed once. Otherwise, it is executed for $\left \lceil T/\tau \right \rceil$ times. We can adjust the setting of $\tau$ to match different FL applications.
Next, we will describe the HMA and the MSA protocols in detail, respectively.

\begin{figure*}
\fbox{%
\parbox{\textwidth}{%
\begin{enumerate}[\textbullet]

\item \textrm{Assumption: $N  $ Clients $\mathcal{C}_N=\{c_1,c_2,...,c_N\}$ jointly train a DNN model whose number of trainable parameters is $M$, under the coordination of the server $A$; the number of masking seed pairs agreed in per Masking Seed Agreement Protocol is $\tau $;}
\item \textrm{While the model has not converged to desired performance:}
\end{enumerate}
\begin{itemize}
   \item[\labelitemii]\textrm{\emph{Masking Seed Agreement Protocol}.\\
   Input: $K_u=\{k_u^1,k_u^2,...,k_u^\tau \}$ from each client $u$.\\
   Output: $K_0=\{k_0^t=\sum_{u\in \mathcal{C}_N} k_u^t: 1\leqslant t \leqslant \tau\}$ to each client.
   \item[\labelitemii] For the training epoch $t = 1,2,...\tau$, do}
   \begin{itemize}
    \item[-] \textrm{\emph{Homomorphic Model Aggregation Protocol}.\\
    Input: the list of model update and masking seed pairs ${\{m_u^t,k_u^t,k_0^t\}}_{u\in \mathcal{C}_N}$from each client $u$.\\
    Output: the aggregation of all clients’ model updates $m_0^t=\sum_{u\in \mathcal{C}_N} m_u^t$ to each client.
    \item[-] If the model converges to the desired performance, end the for-loop.}
  \end{itemize}
\end{itemize}
}%
}
	\caption{The Overview of Doubly Homomorphic Secure Aggregation.}
	\label{FIG:2}
\end{figure*}

\subsubsection{The Homomorphic Model Aggregation Protocol}

In the Homomorphic Model Aggregation Protocol, the aggregation of clients’ local models is computed under the orchestration of the server, ensuring no information about the individual models is revealed beyond their aggregated value. Figure~\ref{FIG:3} shows the HMA protocol over the online clients that construct the set $\mathcal{C}_N$. Each client first locally samples a random masking seed as the input to the SHPRG and stretches it to a mask for all entries of the model update. Then, they upload masked model update $y_u=x_u+G(k_u)$ to the server. The server aggregates the uploaded data of online clients, and broadcasts $y_0=\sum y_u$ to online clients who receive $k_0$ from MSA. The clients remove the mask which is $G(k_0)$ from $y_0$, and dequantize the result before computing the average to obtain the updated global model. Besides, malicious clients combining the uploaded data from different epochs may induce information leakage about the individual client. To address this threat, clients apply different masking seeds in different epochs to obtain disposable masks.

We instantiate the protocol by the almost seed homomorphic PRG introduced in Section 2. Since the output of SHPRG is in $\mathbb{Z}_{p}^{m}$, we set the public modulus $P$ in our scheme equal to $p$, and the model updates are preprocessed by the quantization operation, which converts each bounded local model update to $w$-bit integer before adding masks. For a model update $m$ in $[m_{\textrm{min}}, m_{\textrm{max}})$, the quantized value of $m$ is 
$$
Q(m)=\left\lfloor\frac{2^w(m-m_{\textrm{min}})}{m_{\textrm{max}}-m_{\textrm{min}}}\right\rfloor,
$$
where $\left\lfloor a \right\rfloor$ is the flooring function that maps $a\in\mathbb{R}$ to the largest integer not greater than $a$. The aggregation of quantized value over $N$ parties is at most $N(2^w-1)$, so we set $p>N(2^w-1)$ to make sure the summed model update does not overflow. For summation result $x$, the corresponding dequantization is performed by 
$$
Q^{-1}(x) = 2^{-w}(m_{\textrm{max}}-m_{\textrm{min}})x+Nm_{\textrm{min}}.
$$

\begin{figure*}
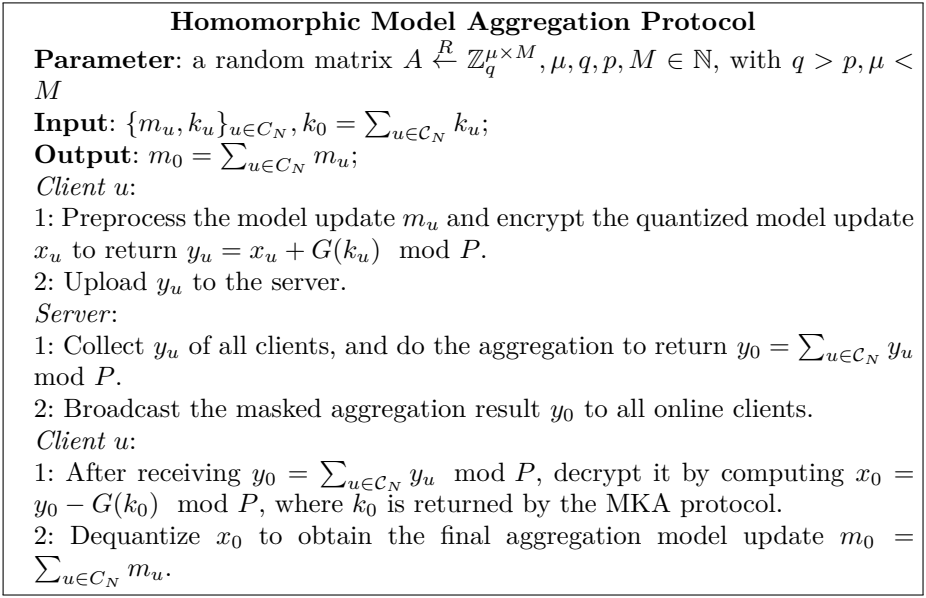

\fbox{%
  \parbox{\textwidth}{%
\begin{center}
\centering{\textbf{Homomorphic Model Aggregation Protocol}}
\begin{enumerate}[-]
\item[]\textbf{Parameter}: a random matrix $A\overset{R}{\leftarrow}\mathbb{Z}_{q}^{\mu \times M},\mu, q, p, M\in\mathbb{N}$, with $q>p, \mu<M$ 
\item[] \textbf{Input}: $\{m_u, k_u\}_{u\in C_N}, k_0=\sum_{u\in \mathcal{C}_N} k_u$;
\item[] \textbf{Output}: $m_0=\sum_{u\in C_N} m_u$;   
\item[] \emph{Client $u$}:\\
\item[]1: Preprocess the model update $m_u$ and encrypt the quantized model update $x_u$ to return $y_u=x_u+G(k_u) \mod P$.
\item[] 2: Upload $y_u$ to the server.
\item[] \emph{Server}:
\item[] 1: Collect $y_u$ of all clients, and do the aggregation to return $y_0=\sum_{u\in \mathcal{C}_N} y_u\mod P$. 
\item[] 2: Broadcast the masked aggregation result $y_0$ to all online clients.
\item[] \emph{Client $u$}:
\item[] 1: After receiving $y_0=\sum_{u\in \mathcal{C}_N} y_u\mod P$, decrypt it by computing $x_0=y_0-G(k_0)\mod P$, where $k_0$ is returned by the MKA protocol.
\item[] 2: Dequantize $x_0$ to obtain the final aggregation model update $m_0=\sum_{u\in C_N} m_u$.
\end{enumerate}
\end{center}
}%
}
\caption{The Homomorphic Model Aggregation Protocol.}	\label{FIG:3}
\end{figure*}
\subsubsection{The Masking Seed Aggregation Protocol}
In the Masking Seed Agreement protocol, the server and clients jointly compute the demasking seeds which are the sum of the masking seeds of corresponding clients, i.e.$k_0=\sum k_u$. The inputs of the MSA protocol are masking seeds for $\tau$ epochs $K_u=\{k_u^1,k_u^2,...,k_u^\tau \}$ from each client $u$, where $k_u^t$ is generated independently by the client $u\in \mathcal{C}_N$ for epoch $t$, and kept private by the individual. The demasking seeds for $\tau$ epochs $K_0=\{k_0^1,k_0^2,...,k_0^\tau \}$ are computed and released to only clients without revealing the individual masking keys, which is achieved by the compact MK-BFV technique. 

As illustrated in Figure~\ref{FIG:4}, the first step of the protocol is key generation to set up the common public key and the re-encryption key. Each client generates a pair of secret and public keys individually, and the common public key is computed for encryption based on the individual pubic keys. Here, the same public parameter is taken as the input of the key-generation algorithm so that the multi-key homomorphic arithmetic is supported. For a client $u$, after the individual key pair $\{\textrm{sk}_u,\textrm{pk}_u\}$ obtained, it uploads the public key $\textrm{pk}_u$ to the server. The server combines all received public keys to return a common public key cpk, and broadcasts it to all the clients for encryption. Also, one leader client generates the re-encryption key pair and releases it to other clients who have access to the final decryption in a secure manner. The Diffie-Hellman Key Agreement technique and Authenticated Encryption can be combined to prevent revealing the re-encryption key pair to any external party including the server \cite{RN44}. It's not the core of our work, so implementation details are omitted here. The re-encryption secret key should be kept private by clients who are the model receiver. To verify the validation of the re-encryption key pair, each client performs encryption and decryption on the received key pair to check the correctness of the decryption result. 

After key generation, the participants perform two rounds of communication for encryption, evaluation, re-encryption and the final decryption. In the first round, each client encrypts the masking seeds with the common public key before sending it to the server. Then, the server aggregates all the received ciphertexts to return the encrypted sum over all masking seeds, which is $\textrm{ct}=\sum{\textrm{ct}_u}$ under the joint secret key $\textrm{sk}=\sum{\textrm{sk}_u}$. In the next round, the clients and the server execute the PubKeySwitch protocol to re-encrypt $\sum{\textrm{ct}_u}$ outputting the ciphertext ReEnc(ct)=Enc(pk$_\textrm{r}$, Dec(sk,ct)) which can be decrypted with the re-encryption secret key $\textrm{sk}_\textrm{r}$. Finally, related clients utilize sk$_\textrm{r}$ to decrypt the received ciphertexts $\textrm{ct}_{\textrm{r}}$ and obtain the demasking seeds.
\begin{figure}[h]%
	\centering
	\includegraphics[width=0.9\textwidth]{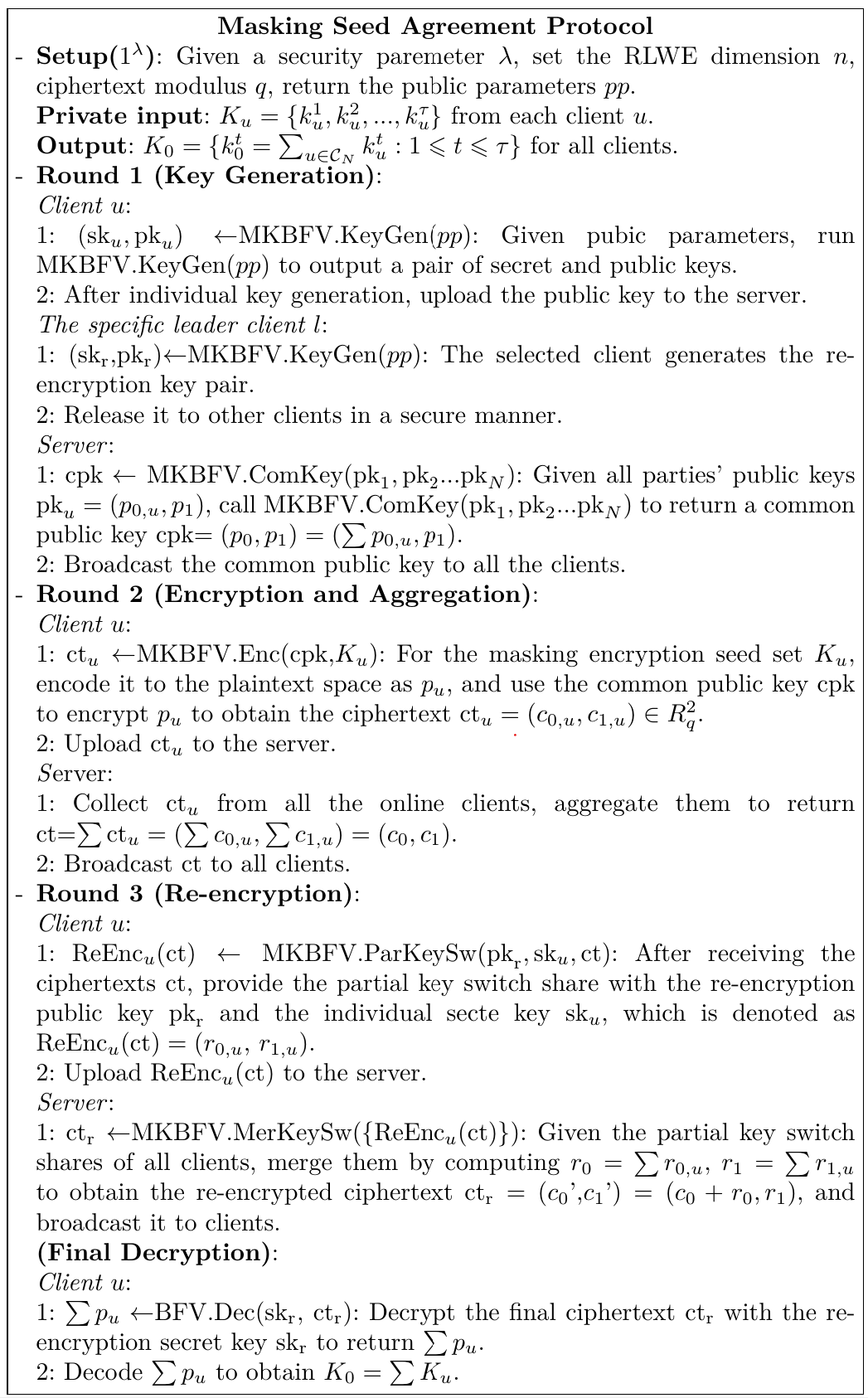}
	\caption{The Masking Seed Aggregation Protocol.}\label{FIG:4}
\end{figure}

\section{Correctness and Security}\label{sec5}
In this section, we state our correctness and security theorems. We consider clients in $\mathcal{C}_N$ and the server $A$ execute DHSA with inputs $m_{\mathcal{C}_N}=\{m_u: u\in \mathcal{C}_N\}$, $\mid \mathcal{C}_N \mid=N$. 
\subsection{Correctness}
Here, we provide a preceding statement of correctness prior to the security analysis.
\begin{theorem}[Correctness Theorem]
If participants in $\mathcal{C}_N$  follow the HMA protocol, given the demasking seed $k_0$ returned by the MSA protocol, clients can obtain $\sum_{u\in\mathcal{C}_N}m_u$ with  negligible noise.
\end{theorem}
\begin{proof}
Because the selected PRG is almost seed-homomorphic, we have:
\begin{equation}
\begin{aligned}
\sum{_{i=1}^n}G(k_i)&=G(\sum{_{i=1}^n}k_i)+e \mod p\\& \textrm{where}, e\in \{-n+1,...,0,1,...,n-1\}
\end{aligned}
\end{equation}
For the HMA protocol, $y_u=m_u+G(k_u)\mod P$, where
$G(k_u)\in \mathbb{Z}_P^M, P=p,p\geqslant N(2^w-1)+1)$, we have:
\begin{equation}
\begin{aligned}
m_0&=\sum_{u\in \mathcal{C}_N}y_u-G(k_0)\mod P\\&=\sum_{u\in \mathcal{C}_N}(m_u+G(k_u))-G(\sum_{u\in \mathcal{C}_N}k_u)\mod P\\&=\sum_{u\in \mathcal{C}_N}m_u+\sum_{u\in \mathcal{C}_N}G(k_u)-G(\sum_{u\in \mathcal{C}_N}k_u)\mod P\\&=\sum_{u\in \mathcal{C}_N}m_u+e_0\mod P,e_0\in\{-N+1,...,0,1,...,N-1\}
\end{aligned}
\end{equation}
\end{proof}
The noise here is insignificant relative to the domain of aggregated quantized model updates which is $N(2^w-1)$, and can be demonstrated to have a negligible impact on the quality of the trained model experimentally. We conclude that the model aggregation can be computed correctly based on the given $k_0$.

\subsection{Security}
Then we show that DHSA achieves the two aspects of security goals set for cross-silo FL as described in Section 2.4. We prove our scheme is secure against the server colluding with up to $N-2$ clients in the honest-but-security setting. Those clients and the server learn nothing more than their own inputs, and the sum of the inputs, masking seeds and masks of the other clients. Note that for the ideal aggregation scheme, if the server corrupts a set of clients, the partial aggregation result of the remaining clients will be disclosed as well. The information obtained by the colluding participants in our scheme is the same as the ideal case. Meanwhile, the security guarantee is established on the setting without a TTP, where intermediate model data is visible to the necessary parties involved.  

We consider the execution of DHSA with privacy threshold $T_{col}=N-2$, and underlying cryptographic primitives are instantiated with security parameters $\Lambda$. In such a secure aggregation execution, the view of a client $u$ consists of its internal state (including its model update $m_u$, masking seed $k_u$, mask $G(k_u)$, demasking seed $k_0$, individual key pair for MK-BFV $\{\textrm{sk}_u, \textrm{pk}_u\}$, the aggregated model update $\sum m_{u}$) and all messages this party received from other parties (including common public key cpk, ciphertexts ct and $\textrm{ct}_\textrm{r}$). The view of the server $A$ consists of the received information including individul public keys $\{ \textrm{pk}_u\}_{u\in \mathcal{C}_N }$, ciphertexts $\{ \textrm{ct}_u\}_{u\in \mathcal{C}_N}$, the partial re-encyption share $\{ \textrm{ReEnc}_u(\textrm{ct})\}_{u\in \mathcal{C}_N}$ and the masked model updates $\{y_u\}_{u\in \mathcal{C}_N}$. The messages sent by this party will not be part of the view because they can be determined using the other elements of its view.  

Given any subset $\mathcal{V}\subset \mathcal{C}_N\cup A$, let $\textsf{REAL}_{\mathcal{V}}^{\mathcal{C}_N,T_{col},\Lambda}$ be a random variable representing the combined views of all parties in $\mathcal{V}$ in the execution of DHSA, where the randomness is over the internal randomness of all parties, and the randomness in the setup phase. We show that for any such set $\mathcal{V}$ of honest-but-curious clients of size up to $N-2$, the joint view of $\mathcal{V}$ can be
simulated given the inputs of the clients in $\mathcal{V}$, and the sum of the inputs, masking seeds and masks of the other clients.
\begin{theorem}[Security Theorem]
 There exists a probabilistic polynomial
time (PPT) simulator $\textsf{SIM}$ such that for all $\mathcal{C}_N,m_{\mathcal{C}_N}$, and $\mathcal{V}\subset \mathcal{C}_N\cup A, \mid \mathcal{V}\setminus\mathcal{A} \mid <N-1 $, the output of SIM is computationally indistinguishable from the joint view of $\textsf{REAL}_{\mathcal{V}}^{\mathcal{C}_N,T_{col},\Lambda}$ of the parties in $\mathcal{V}$:
\begin{equation}
\begin{aligned}
&\textsf{REAL}_{\mathcal{V}}^{\mathcal{C}_N, T_{col}, \Lambda}(m_{\mathcal{C}_N}, \mathcal{C}_N)\approx \textsf{SIM}_{\mathcal{V}}^{\mathcal{C}_N,  T_{col}, \Lambda}(m_\mathcal{V}, z_m, z_k, z_g, \mathcal{C}_N)\\
&z_m=\sum_{u\in\mathcal{C}_N\setminus \mathcal{V}}m_u,z_k=\sum_{u\in \mathcal{C}_N\setminus \mathcal{V}}{k_u},z_g=\sum_{u\in\mathcal{C}_N\setminus \mathcal{V}}{G(k_u)}  
\end{aligned}
\end{equation}
\end{theorem}
\begin{proof}
We prove the theorem by a standard hybrid argument. We 
will present a series of hybrids from variable \textsf{REAL} to \textsf{SIM} where any two subsequent random variables are computationally 
indistinguishable. We assume that $A\in \mathcal{V}$, which indicates the view of the server should be considered. The case where $A$ is not in $\mathcal{V}$ is much easier to prove and is omitted for brevity.
\begin{itemize}
\item[$\textsf{Hyb}_0$] In this hybrid, the variables are distributed exactly as in \textsf{REAL}. We choose a specific client ${u}'$ in $\mathcal{C}_N\setminus \mathcal{V}$. For this client, based on the given $z_m$, $z_k$ and $z_g$, we can write as $y_{{u}'}=m_{{u}'}+G(k_{{u}'})=z_m+z_g-\sum_{u\in\mathcal{C}_N\setminus \mathcal{V}\setminus \{{u}'\}}y_u$, $k_{{u}'}=z_k-\sum_{u\in\mathcal{C}_N\setminus \mathcal{V}\setminus \{{u}'\}}k_u$.
\item[$\textsf{Hyb}_1$] In this hybrid, for a party $u$ in $\mathcal{C}_N\setminus \mathcal{V}\setminus \{u'\}$, in HMA protocol instead of sending $y_u=m_u+G(k_u)$, we send $y_u=m_u+P_u$, where $P_u$ is uniformly random. For ${u}'$, the masked data is still generated by $y_{{u}'}=z_m+z_g-\sum_{u\in\mathcal{C}_N\setminus \mathcal{V}\setminus \{u'\}}y_u$. The security of SHPRG guarantees that the distribution of $\{y_u: u\in\mathcal{C}_N\setminus \mathcal{V}\setminus \{{u}'\}\}$ is identically distributed to the corresponding one in $\textsf{Hyb}_0$. On the other hand, $y_{{u}'}$ is determined by \{${y_u}:u\in \mathcal{C}_N\setminus \mathcal{V}\setminus \{{u}'\}\}$, $z_m$ and $z_g$,  so the distribution of \{${y_u}: u\in \mathcal{C}_N\setminus \mathcal{V}\}$ is identically distributed to that in $\mathsf{Hyb}_0$.
\item[$\textsf{Hyb}_2$] In this hybrid, for party $u$ in $\mathcal{C}_N\setminus \mathcal{V}\setminus \{{u}'\}$, we replace the uploaded data in HMA protocol by $y_u=P_u$, which is possible since $P_u$ was obtained in $\textsf{Hyb}_1$ to be uniformly random, $m_u+P_u$ is also uniformly random. For the chosen client ${u}'$, its uploaded data is still computed by $y_{{u}'}=z_m+z_k-\sum_{u\in\mathcal{C}_N\setminus \mathcal{V}\setminus \{{u}'\}}y_u$,  which makes the joint view of clients in $\mathcal{C}_N\setminus \mathcal{V}$ consistent with the previous one, and the joint distribution of the data uploaded by clients in $\mathcal{C}_N$ stays identical. Hence the joint view of the participants including the server is indistinguishable from the previous hybrid.

\item[$\textsf{Hyb}_3$]  In this hybrid, for a party $u$ in $\mathcal{C}_N\setminus \mathcal{V}\setminus \{u'\}$, in Encryption and Aggregation step of the MSA protocol, instead of sending $\textrm{BFV.Enc}(k_u)$, we send $\textrm{BFV.Enc}({P_u}')$, where ${P_u}'$ is uniformly random. Based on the security of BFV, the ciphertexts are distributed identically. For $u'$, $k_{{u}'}=z_{k}-\sum_{u\in\mathcal{C}_N\setminus \mathcal{V}\setminus \{{u}'\}}P_u'$, which guarantees that the joint distribution of the data uploaded by clients in $\mathcal{C}_N$ stays identical. The masking seeds uploaded by clients in the MSA protocol sum up to $k_0$ staying identical to the previous one. 

\item[$\textsf{Hyb}_4$] In Re-encryption step of the MSA protocol, for a party $v$ in $\mathcal{V}$, send $(r_{0,v}, r_{1,v})=(s_vc_1+u_vp_0'+e_{0,v}, u_vp_1'+e_{1,v})$; for a party $u$ in $\mathcal{C}_N\setminus \mathcal{V}\setminus \{{u}'\}$, replace the uploaded data by $(r_{0,u}, r_{1,u})=(a_{0,u}, a_{1,u})$ where $a_{0,u}, a_{1,u}$ are uniformly random in $R_q$. The usual RLWE assumption suffices the replacement is indistinguishable. For the specific party ${u}'$, shared data $(r_{0,{u}'}, r_{1,{u}'})=(a_{0,{u}'}, a_{1,{u}'})$, and $a_{0,{u}'}, a_{1,{u}'}$ are computed by $ a_{0,{u}'}={c_0}'-c_0-\sum_{u\in{\mathcal{C}_N}\setminus{\mathcal{V}}\setminus \{{u}'\}} {r_{0,u}}-\sum_{v\in\mathcal{V}}{h_{0,v}}$,
$ a_{1,{u}'}={c_1}'-\sum_{u\in{\mathcal{C}_N}\setminus{\mathcal{V}}\setminus \{{u}'\}} {r_{1,u}}-\sum_{v\in\mathcal{V}}{h_{1,v}}$

where $\textrm{ct}_{\textrm{r}}=({c_0}',{c_1}')$ is a public output. By doing this, the distribution of  $\{(r_{0,u}, r_{1,u}): u\in \mathcal{C}_N\}$ is identically distributed to that in $\textsf{Hyb}_3$, and also the simulated key switch shares sum up to $\textrm{ct}_{\textrm{r}}$, making sure the joint view of the outputs of other rounds identical to the previous hybrid.

\item[$\textsf{Hyb}_5$] In this hybrid, in Key Generation round of MSA protocol, cpk=$(p_0,p_1)$ is the output, and its transcript is the tuple $(p_{0,1},p_{0,2},...p_{0,N})$ of all the parties' shares. In our simulator, for a party $v$ in $\mathcal{V}$, $\textrm{pk}_v=(p_{0,v},p_1)$ is generated by $p_{0,v}=-s_v\cdot p_1+e_v$. For a party $u$ in $\mathcal{C}_N\setminus{\mathcal{V}}\setminus \{{u}'\}$, we substitute the locally generated public key $\textrm{pk}_u$ by $(p_{0,u},p_1)$, where $p_{0,u}$ is uniformly random from $R_q$. For client ${u}'$, we compute $p_{0,{u}'}= p_0-\sum_{u\in{\mathcal{C}_N}\setminus{\mathcal{V}}\setminus \{{u}'\}} {p_{0,u}}-\sum_{v\in\mathcal{V}}{p_{0,v}}$ and send $\textrm{pk}_{{u}'}=(p_{0,{u}'},p_1)$ to the server. Here, $\textrm{cpk}=(p_0,p_1)$ is the public output. Thus, the simulated shares sum up to $p_0$ which means the output cpk of this round is equal to the real ones, and the distribution is indistinguishable from that of the previous hybrid. The view of other rounds of MSA protocol results in valid keys and ciphertexts of the BFV scheme, which indeed preserves its semantic security. Also, the final decryption can be performed successfully and output the real result, preventing the adversary from distinguishing it.  
\end{itemize}
\end{proof}

Thus, the PPT simulator \textsf{SIM} that samples from the distribution described in the last hybrid can output computationally indistinguishable from \textsf{REAL}, the distribution can be computed based on $m_{\mathcal{V}},z_m,z_k,z_g$. The joint view of up to $N-2$ clients can be simulated without learning individual data, which means our scheme can preserve the security against the aggregator colluding with an arbitrary subset of up to $N-2$ clients. 

Besides, the masking seeds in the HMA protocol and the encryption keys in the MSA protocol are generated locally or collaboratively computed among clients. Our scheme provides a strong security guarantee without the need for a TTP.

As a side note, the Public Key Switch protocol achieves the goal to release the intermediate global model and trained model to clients only. Although one concern is that the re-encryption keys $(\textrm{sk}_{\textrm{r}}, \textrm{pk}_{\textrm{r}})$ for the PKS protocol in MSA introduces secretly shared information among clients, which may be disclosed to the server in colluding case, it does not violate the security goal because the re-encryption key can only be used to decryption the re-encryption ciphertexts $\textrm{ct}_{\textrm{r}}$. In the Re-encryption step, to collaboratively compute $\textrm{ct}_{\textrm{r}}$, all the clients need to upload the partial key switch share $\textrm{ReEnc}_i(\textrm{ct})$, which is computed upon the aggregated ciphertexts ct, to the server for merging. In honest-but-curious setting, if the colluding adversaries learn $\textrm{sk}_{\textrm{r}}$ to decrypt $\textrm{ct}_{\textrm{r}}$, only the aggregation result will be revealed. Even in the ideal aggregation scheme, as soon as the clients obtain the plain aggregation results, the aggregation results will be disclosed if the server corrupts any client. Thus, the information exposure at this level is not a limitation of PKS protocol. In comparison, for the previous scheme based on conventional single-key HE, the individual model update will be exposed to someone else for all ciphertexts can be decrypted with the same secret key. In colluding cases, the clients may learn other clients' ciphertexts, or the server may learn the secret key, which enables an illegal decryption. DHSA enables the final decryption done by clients, meanwhile avoiding the exposure of individual model updates against colluding. 

In all, DHSA achieves the security goals described in Section 2.4, and provides a stronger privacy guarantee compared with previous solutions, as shown in Tables~\ref{tab1}. 
\begin{sidewaystable}
\sidewaystablefn%
\begin{center}
\begin{minipage}{\textheight}
\caption{Comparison of privacy-preserving approaches in FL systems.}\label{tab1}
\begin{tabular*}{\textheight}{@{\extracolsep{\fill}}lcccccccc@{\extracolsep{\fill}}}
\toprule%
 \multicolumn{3}{c}{}&Double-mask & BatchCrypt & TP & POISONDON & Ours\\
  \cline{1-8}
  &\multicolumn{2}{c}{$T_{col}$\footnotemark[1]}&$\alpha N (\alpha \in (0,1))$&$\times $&$\alpha N (\alpha \in (0,1))$&$N-2$&$N-2$\\
  \cline{2-8}
  \textbf{Security}&\multicolumn{2}{c}{Global model released}&Server&clients&server&querier\footnotemark[2]&clients\\
  \cline{2-8}
  &\multicolumn{2}{c}{Non-assumption of TTP}&\checkmark&\checkmark&$\times$&\checkmark&\checkmark\\
  \cline{1-8}
  &\multirow{2}{*}{\makecell[c]{Communication \\($T$ epochs)}}&Rounds&$4T$&$T$&$4T$&2T+1&$T+3\left \lceil T/\tau \right \rceil$ \\
  \cline{3-8}
  \multirow{3}{*}{\textbf{Efficiency}}&&\centering{\makecell[c]{Traffic \\inflation}}&$1\sim 2$&$>2$& $>100$&$>20$&$1\sim 2$\\
  \cmidrule{2-8}
  &\multicolumn{2}{c}{Computation($T$ epochs)}&$O(MN^2)$&$O(M)$&$O(M)$&$O(M)$&$O(M)$\\
\botrule
\end{tabular*}
\footnotetext[1]{$T_{\textrm{col}}$ of Double-mask and TP methods is determined by the parameterized threshold setting $t=\alpha N$, which means the decryption of the aggregated value needs the query of at least $t$ data parties.}
\footnotetext[2]{A querier in POISONDON can be one of the $N$ parties or an external entity – queries
the model and obtains prediction results on its evaluation data.}
\end{minipage}
\end{center}
\end{sidewaystable}

\subsection{Experimental Settings}
\begin{enumerate}[\textbullet]
\item Benchmarking Models. We implement three representative machine learning applications in FL, and perform plain aggregation and DHSA for each one. Our first application is a CNN model consisting of two convolutional layers with a total of about 0.2M parameters, trained over the FashionMNIST dataset \cite{mnist}. In another application, we train ResNet18 \cite{DBLP:journals/corr/HeZRS15} with 11M parameters on the CIFAR10 dataset \cite{2012Learning}. In the third application, we use Shakespeare dataset \cite{rnn} to train a customized LSTM \cite{hochreiter1997long} with 1.25M parameters. The three applications are based on different types of machine learning models of various sizes, and cover the learning tasks for image classification and text generation. The optimization approach for federated learning is the Federated Averaging algorithm \cite{mcmahan2017communication}. For plain FedAvg aggregation, the model updates are represented by real-valued vectors of 32 bits and uploaded for aggregation without encryption. For DHSA, before being encrypted, the bounded model weights are quantized into 16-bit unsigned integers, i.e. $w=16$. 
\item Homomorphic Model Aggregation Protocol Implementation. In our implementation, we set the baseline setting of parameters of SHPRG used in the HMA protocol as $\mu=512, p=2^{24}, q=2^{54}$, and the LWE evaluator estimates a hardness of over $2^{233}$ \cite{RN76}. Also, $q/p>\mu$, which ensures the LWR problem appears to be exponentially hard for any p=poly($\lambda$) where $\lambda$ is the security parameter \cite{RN74}.  
\item Masking Seed Aggregation Protocol Implementation. We set the parameters of the MK-BFV as $n=2^{12}, \log_2q=109, \log_2t=64$ which is with 192-bit security \cite{albrecht2021homomorphic}. The 64-bit $t$ (packing-compatible) can cover the computation domain of the masking seed agreement process. For the running time examination, we set the number of the masking seed pairs agreed during per execution of MSA $\tau$ as 100. Note that we can select larger $\tau$ to gain higher efficiency.

All experiments are run in a Lenovo server with the configuration of Ubuntu 20.04, Intel(R) Core i7-10700K@3.80GHz CPU$\times$16 and NVIDIA GeForce RTX2080 SUPER. We implement DHSA in Go, which builds on top of Lattigo \cite{lattigo}, an open-source Go library for lattice-based cryptography.
\end{enumerate}

\subsection{Efficiency of DHSA}
We evaluate the efficiency from the aspect of communication and computation overheads, and compare the results with two state of the art methods, BatchCrypt \cite{RN48} and POISONDON \cite{DBLP:conf/ndss/SavPTFBSH21}. BatchCrypt optimizes the efficiency of HE-based secure aggregation solution for cross-silo FL, and POISONDON utilizes multi-key CKKS which provides the same security guarantee with our method. We set the parameters of BatchCrypt the same as that in their paper, and set the the parameter of POSEIDON as $n=2^{12}, \log_2q=109, \sigma = 3.2$, which ensures 128 bits security without too much redundancy in computation and communication. All calculations below assume a single server and $N$ clients, where each client holds a model whose size is $M$.
\subsubsection{Computation Efficiency}

We first state the analysis of the asymptotic computation efficiency of each protocol, and then conduct it with the running time of practical execution experimentally. We measure the running time of two protocols separately and the total running time of computation for one single epoch. We execute the tests 100 times and take the average. 

In the HMA protocol, the computation cost is mainly derived from computing SHPRGs to generate masks for each entry in the model update vector. The computation costs for each client and the server are $O(M)$. We then test the running time for different model sizes. Since the selection of parameters influences the security, the number of clients the system can handle, and also the efficiency, we test the running time under different setting of parameters of SHPRG. Synthesized vectors are used for locally trained models whose elements are encoded to 16-bit unsigned integers, and the local training time is not included in the total running time. The selected parameter settings and the corresponding running time of each client are shown in Tables~\ref{tab2}. The computation overhead here is linear related to the model size, independent of the number of clients. The running time of the server of HMA is negligible and almost the same as plain FedAvg, for simple aggregation is performed on masked model data whose size stay similar to the raw model data.

The involved parameters of the HMA protocol include $\mu, p, q$. We can see from Tables~\ref{tab2} that when we increase $\mu$, the security level will be enhanced. However, the computation overhead becomes larger. In addition, larger $p$ makes the system able to handle more clients while resulting in higher communication overhead. Given $\mu$, either increasing $p$ or decreasing $q$ will lower the security level. In this paper, we select the most moderate set of parameters to conduct the following efficiency evaluation. We set the parameter to ensure sufficient security guarantee and avoid overflowing when the number of clients is below the threshold. 

\begin{table}[h]
\begin{center}
\caption{Running time of the HMA protocol with different model sizes under different settings of SHPRG parameters.}\label{tab2}%
\begin{tabular}{@{}ccccccccc@{}}
\toprule
& \multicolumn{3}{@{}c@{}}{parameters of SHPRG}&\multirow{2}{*}{security}&\multirow{2}{*}{\makecell[c]{the maximum \\number of clients}}&\multicolumn{3}{@{}c@{}}{\makecell[c]{running time for\\ different model sizes (ms)}}\\
\cmidrule{2-4} \cmidrule{7-9}
&$\mu$&$p$&$q$&&&10K&100K&1M \\
\midrule
Setting A & 512 & $2^{24}$&$2^{54}$&$2^{233}$&256&6&62&624\\
\midrule
Setting B & 512 & $2^{32}$&$2^{64}$&$2^{128}$&65536&6&62&624\\
\midrule
Setting C & 256 & $2^{24}$&$2^{72}$&$2^{132}$&256&3&33&328\\
\midrule
Setting D & 1024 & $2^{32}$&$2^{48}$&$2^{244}$&65536&13&129&1291\\
\botrule
\end{tabular}
\end{center}
\end{table}
In the MSA protocol, MK-BFV is utilized to encrypt $\tau$ masking seeds, where the length of each one is $\mu$, and the computation cost of each client is composed of four components: generating encryption keys, encrypting the messages, computing the partial re-encryption share and decryption. For the server, the computation cost can be broken down into evaluating the ciphertexts (addition operation in our protocol) and merging the partial re-encryption shares.  The running time of the MSA protocol for each phase is evaluated, and the results of each step are listed in Tables~\ref{tab3}. Note that the computation overhead of clients here does not depend on the amount of the model parameters and the number of clients. Due to the small data size needed to be exchanged, the overhead is relatively small. What’s more, the MSA protocol runs once per $\tau$ epochs during the overall training process, whose average overhead is negligible. It reduces the number of communication rounds, and can make full use of BFV’s packing technology to improve efficiency.

\begin{table}[h]
\begin{center}
\begin{minipage}{250pt}
\caption{Running time of the MSA protocol (ms).}\label{tab3}
\begin{tabular}{@{} cccccc@{} }
\toprule
&KeyGen&Enc\&Agg&PubKeySwi&Final Dec&Total\\
\midrule
client&0.27&10.79&11.96&3.64&26.66\\
\midrule
server&0.07&0.8&1.7&--&2.57\\
\botrule
\end{tabular}
\end{minipage}
\end{center}
\end{table}
To evaluate the total computation overhead resulting from DHSA, we compare the running time of computation for an integral iteration with plain FedAvg, BatchCrypt and POSEIDON. In the plain FedAvg, the learning process without secure aggregation, clients train the model for $t_{local}$ epochs before the plain model updates are aggregated by the server party. The running time here includes the time of both local training and aggregation. For other secure aggregation methods, an integral iteration includes $t_{local}$ epochs of training, encryption of the model update, aggregation of ciphertexts, and decryption of the result. The total running time is evaluated. We set $N=10, \tau=100, t_{local}=10$ for the FL setting applied for 2-layer CNN, ResNet18 and LSTM. As the results visualized in Figure~\ref{FIG:5} show, our DHSA reduces the computation overhead significantly. The extra running time over the original FedAvg owing to DHSA is insignificant. Compared with BatchCrypt, DHSA provides up to 20× speedup for aggregation. Specifically, when the size of the machine learning model increases, the computation overhead of DHSA increases more gently. DHSA gains a greater speedup over the baseline when the model size gets larger.  
\subsubsection{Communication Overhead}

The communication traffic mainly comes from the HMA protocol, where the masked models of size $M\log_2p$ are uploaded. The communication cost is $O(M)$ for each client and $O(MN)$ for the server, equal to the plain learning of FL.  

For the MSA protocol, the inputs are the masking seeds whose length is parameterized, resulting in a constant communication overhead depending on the selected parameters of MK-BFV. Based on parameters $n$ and $q$ we select for MK-BFV, the public key and ciphertext size is about $2n\log_2q$. The overhead includes three rounds of communication. The common public key is agreed in the first round, where each client uploads the individual public key 
whose size is about $2n\log_2q$, and the server broadcasts the common public key of the same size. The information exchanged in the second round is the uploaded ciphertexts and the sum of the ciphertexts broadcast. For each client, the message that needs to be encrypted is the masking seed, a private vector of dimension $\mu$. For $\tau$ pairs of masking seeds, the number of total messages is $\mu\tau$. With the packing encryption method, the messages can be encoded into $\left \lceil \mu\tau/n \right \rceil$ plaintexts, and the size of corresponding ciphertexts for each client is $2n\left \lceil \mu\tau/n \right \rceil\log_2q$. In the third round, clients upload the partial re-encryption share, and the re-encryption ciphertexts are returned, both with the size of $2n\left \lceil \mu\tau/n \right \rceil\log_2q$ approximately.  

Figure~\ref{FIG:6} depicts the communication comparison in terms of traffic volume. We see that our scheme reduces the amount of communication traffic compared with BatchCrypt and POSEIDON, and the traffic inflation factor is reduced to approximately 1.5. We define the traffic inflation factor as the ratio of communication traffic (between each client and the server) of secure aggregation and plain FedAvg. Although POSEIDON inproves the computation overhead compared with BatchCrypt, the communication overhead is the bottleneck. In addition, we stress that we can further cut down the communication traffic by adjusting the parameters of SHPRG, e.g. smaller $p$ which tolerates fewer participating clients. For FL system with 10 clients, if we set $p=2^{20}$, the inflation factor can be reduced to 1.25.
\begin{figure}
	\centering
	\includegraphics[scale=0.4]{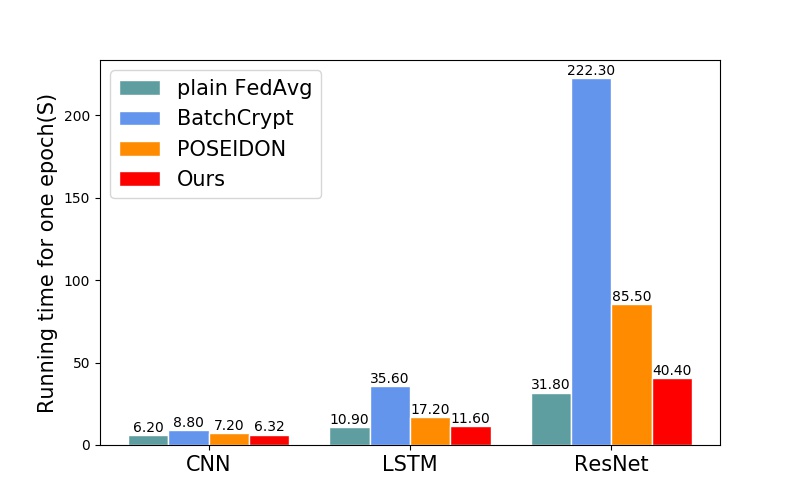}
	\caption{Total running time of computation for one epoch.}
	\label{FIG:5}
\end{figure}

\begin{figure}
	\centering
	\includegraphics[scale=0.4]{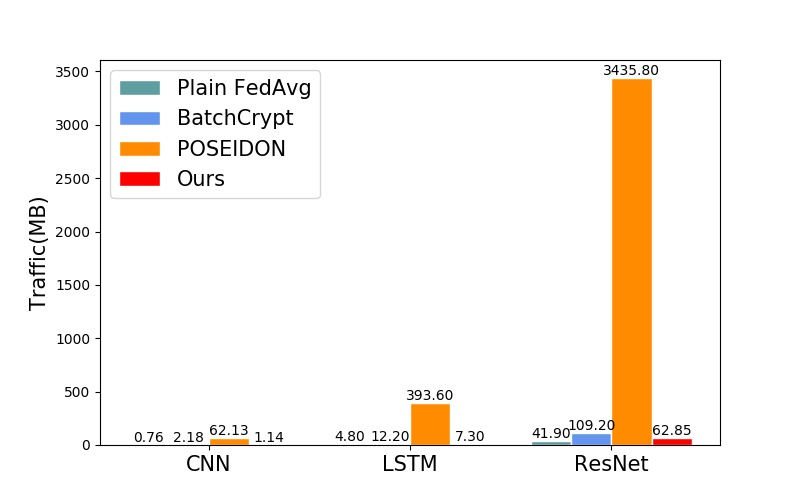}
	\caption{The Communication Overhead.}
	\label{FIG:6}
\end{figure}

\subsection{Quality of Trained Model}
For our scheme, the model updates have two sources of error: (1) the model parameters are quantized into 16-bit integers before masking, and corresponding dequantization is done after 
aggregation; (2) SHPRG induces an error term to aggregated 
model parameters. To measure the model quality, we track the test accuracy for CNN and ResNet18. Training loss is used for LSTM
as the dataset is unlabelled and has no test set. As Figure~\ref{fig:res} shows, for one thing, compared with plain FedAvg, the convergence achieves after training for the same epochs, which means the speed of convergence is not affected. For another thing, the trained models obtained by our scheme reach the same peak accuracy or bottom loss as the plain FedAvg.

\begin{figure}[H]

\begin{tabular}{ccc}

\begin{minipage}{0.3\linewidth}

  \centerline{\includegraphics[width=3.5cm]{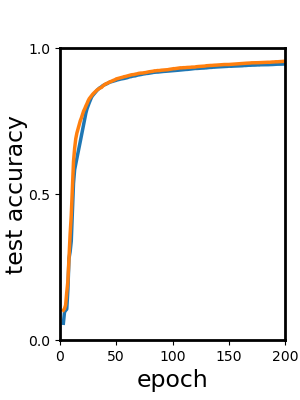}}

  \centerline{(a) CNN}

\end{minipage}

\hfill

\begin{minipage}{0.3\linewidth}

  \centerline{\includegraphics[width=3.5cm]{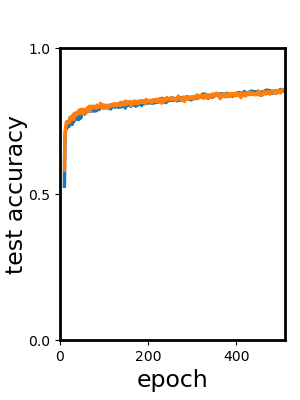}}

  \centerline{(b) ResNet}

\end{minipage}

\hfill

\begin{minipage}{0.3\linewidth}

  \centerline{\includegraphics[width=3.5cm]{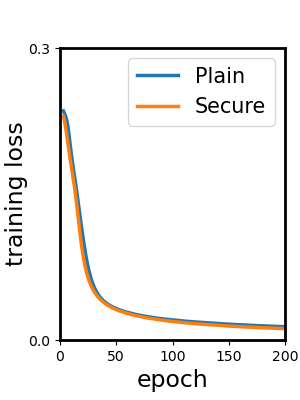}}

  \centerline{(c) LSTM}

\end{minipage}

\end{tabular}

\caption{The quality of trained model.}

\label{fig:res}

\end{figure}

\section{Conclusion}
This paper presents a doubly homomorphic secure aggregation scheme for cross-silo FL settings. We present the application of MKHE which can achieve the stringent security goal of cross-silo FL against colluding parties at a maximum degree. To overcome the bottleneck of high overhead, we utilize SHPRG to reduce the volume of data that needs to be computed securely via MKHE. Overall, we construct a practical privacy-preserving aggregation solution combining the SHPRG-based Homomorphic Model Aggregation protocol and MK-BFV-based Masking Seed Agreement protocol, which is demonstrated to achieve the requirements of a practical and secure cross-silo FL system from security, efficiency and accuracy aspects. 

The privacy security provided by our scheme is robust to the server colluding with up to $N-2$ clients, which can be reached without the need for TTP. Our solution improves the computation efficiency up to 20× over baseline experimentally, and the communication overhead is reduced significantly meanwhile. Additionally, the security and efficiency are provided at no cost of the accuracy of the global model, enabling the practicality for industry deployment.

\bmhead{Acknowledgments}

This paper is supported in part by the National Key
Research and Development Program of China under grant
No.2020YFB1600201, National Natural Science Founda-
tion of China (NSFC) under grant No.(U20A20202, 62090024,
61876173), and Youth Innovation Promotion Association
CAS.

\bmhead{Data Availability Statement}
The datasets generated during and/or analysed during the current study are available from the corresponding author on reasonable request.
\bmhead{Declaration of competing interest}
The authors declare that they have no known competing financial interests or personal relationships that could have appeared to influence the work reported in this article.

\begin{appendices}

\section{Ring Learning With Errors}
For a power-of-two integer $n$ and $R=\mathbb{Z}[X]/(X^n+1)$, define $R_q=R/(q\cdot R)$ as the residue ring of $R$ modulo an integer $q$. The Ring Learning with Errors(RLWE) distribution consists of tuples $(a_i, b_i=s\cdot a_i+e_i)\in R_q^2$, where $s$ is a fixed secret chosen from the key distribution $\chi $ over $R$, $a_i$ is uniformly random in $R_q$, and $e_i$ is an error term drawn from the error distribution $\psi$ over $R_q$. The search RLWE problem states that, given many samples of the form $(a_i, b_i=s\cdot a_i+e_i)\in R_q^2$, it is computationally infeasible to compute the secret $s$. 
\section{BFV}
Here, we detail the  common instantiation of  the basic Brakerski-Fan-\\Vercauteren (BFV) scheme where the ciphertext space is $R_q$, and the plaintext space is the ring $R_t$ for $t<q$ with $\Delta=\left\lfloor q/t \right\rfloor$. The implemention consists of a tuple of algorithms(KeyGen, Enc, Dec, Eval) as below:
\begin{enumerate}
\item[\textbullet]Setup: $pp \leftarrow$ Setup($1^\lambda$): For a given security parameter $\lambda$, set the RLWE dimension $n$, ciphertext modulus $q$, key distribution $\chi$ and error distribution $\psi$. Generate a random vector $a\leftarrow U(R_q)$. Setup($1^\lambda$) returns the public parameter $pp=(n,q,\chi,\psi,a)$. 
\item[\textbullet] Key Generation: $\{\textrm{sk},\textrm{pk}\}\leftarrow$ KeyGen($pp$): Given the public parameter $pp$, KeyGen($pp$) outputs the secret key sk and the public key pk. The secret key is sampled randomly, which is $\textrm{sk}=s\leftarrow \chi$. The public key is set as $\textrm{pk}=(b,a)$, where for the sampled error vector $e\leftarrow \psi$, $b=-s\cdot a+e(\mod q) \in R_q$.
\item[\textbullet] Encryption: ct $\leftarrow$ BFV.Enc($\textrm{pk}, m$): For massage $m\in R_t$, BFV.Enc encrypts it as $\textrm{ct}=(\Delta m+ub+e_0, ua+e_1)$, where $u$ is randomly sampled from $\chi$ and $e_0,e_1$ are sampled from $\psi$. 
\item[\textbullet] Decryption: $m\leftarrow$ BFV.Dec(sk,ct): Taking the secret key sk=$s$ and a ciphertext ct=($c_0, c_1$) as input, BFV.Dec computes $m=\left[\left\lfloor \frac{t}{q}[c_0+c_1s]_q\right\rceil\right]_t$ which is the plaintext corresponding to ct.
\item[\textbullet] Evaluation: ${\textrm{ct}}'\leftarrow$BFV.Eval(pk, $\textrm{ct}_1$, $\textrm{ct}_2$, $f$): Given the ciphertexts $\textrm{ct}_1$, $\textrm{ct}_2$ corresponding to public key pk, as well as the funtion $f$, BFV.Val outputs the ciphertext ${\textrm{ct}}'$ such that BFV.Dec(sk, ${\textrm{ct}}'$)=$f(m_1,m_2)$, where $\textrm{ct}_i$=BFV.Enc(pk, $m_i$).
\end{enumerate}

\section{Multi-key BFV Based On Ciphertexts Extension}
A Multi-key BFV based on ciphertexts extension is another method to  handle
homomorphic computations on ciphertexts under independently generated secret keys. Different from the compact MKBFV, ciphertexts of this scheme are associated to $k$ different parties. The ciphertext is of the form $\overline{\textrm{ct}}=(c_0,c_1,...,c_k)\in R_q^{k+1}$ for a modulus $q$, which is decryptable by the concatenated secret key $\overline{\textrm{sk}}=(1,s_1,...,s_k)$. To achieve the purpose, a key step is the common pre-processing when performing a homomorphic operation between ciphertexts. For given ciphertexts $\textrm{ct}_i=(c_{0,i},c_{1,i})\in R_q^2$ of client $i$, the extended ciphertexts corresponding to the tuple of parties $ (1,2,...N)$ are $\overline{\textrm{ct}_i}=(c_{0,i}^*, c_{1,i}^*,...,c_{N,i}^*)\in R_q^{N+1}$, where $c_{0,i}^* =c_{0,i},c_{j,i}^*=\delta_{ij}c_{1,i}$, and $\delta_{ij}=\left\{\begin{matrix}
1,&\textrm{if }j=i\phantom{-}\\ 
0,&\textrm{otherwise}
\end{matrix}\right.$. Thus, $\overline{\textrm{ct}_i}$ can be decrypted with the joint secret key $\overline{\textrm{sk}}=(1,s_1,...s_N)$. 
For a set of $N$ parties $\mathcal{P}$, this version of MKBFV consists of five PPT algorithms (Setup, KeyGen, Enc, Dec, Eval).
\begin{enumerate}
\item[\textbullet] Setup: $pp \leftarrow \textrm{MKBFV.Setup}(\lambda, \kappa)$. Taking the security and homomorphic capacity parameters as inputs, MKBFV.Setup outputs the public
parameter $pp=\{n,q,\chi,\psi,a\}$. 
\item[\textbullet] Key Generation: $\{\textrm{sk}_i, \textrm{pk}_i\}\leftarrow\textrm{MKBFV.KeyGen}(pp)$. Each party $P_i\in \mathcal{P}$ generates secret and public keys $\{\textrm{sk}, \textrm{pk}\}$ following BFV.KeyGen($pp$). 
\item[\textbullet] Encryption:  ct$_i\leftarrow$MKBFV.Enc($\textrm{pk}_i$, $x_i$). The usual encryption calculation of BFV is used to encrypt message under sk$_i$ to return $\textrm{ct}_i=\textrm{BFV.Enc(pk}_i, x_i)\in R_q^2$.
\item[\textbullet] Evaluation: $\overline{{\textrm{ct}}'}\leftarrow \textrm{MKBFV.Eval}(F,(\textrm{ct}_1,\textrm{ct}_2,...,\textrm{ct}_N), \{\textrm{pk}_i\}_{i\in \mathcal{P}})$. Given a funcion $F$, a tuple of ciphertexts $\textrm{ct}_i=\textrm{BFV.Enc(pk}_i, x_i)=(c_{0,i}, c_{1,i})\in R_q^2$ and the corresponding set of public keys $\{\textrm{pk}_i\}_{i\in \mathcal{P}}$, MKBFV.Eval first extends each ciphertexts $\textrm{ct}_i$ to $\overline{\textrm{ct}_i}\in R_q^{N+1}$ on the joint secret key of set $\mathcal{P}$. Then the arithmetic $F$ is performed on the extended ciphertexts to return $\overline{{\textrm{ct}}'}\in R_q^{N+1} $ . 
\item[\textbullet] Decryption: $x\leftarrow$ MKBFV.Dec($\overline{\textrm{ct}}$, $\{\textrm{sk}_i\}_{i\in \mathcal{P}}$). Given a ciphertext $\overline{\textrm{ct}}$ encrypting $x$ and the corresponding sequence of secret key, MKBFV.Dec outputs the plaintext $x$ by calculating $\left \langle \overline{\textrm{ct}},(1,s_1,...,s_N) \right \rangle$,  where we denote $\left \langle  u,v \right \rangle$ as the usual dot product of two vectors $u,v$.

\end{enumerate} 
\end{appendices}


\bibliography{sn-bibliography}


\end{document}